\documentclass[11pt]{article}
\usepackage{amsmath,amsthm,amssymb,color,verbatim,graphicx,fullpage,url,cite,cleveref}
\newcommand{\remove}[1]{}
\newtheorem{theorem}{Theorem}[section]
\newtheorem{claim}[theorem]{Claim}
\newtheorem{lemma}[theorem]{Lemma}

\newtheorem{definition}[theorem]{Definition}
\newtheorem{corollary}[theorem]{Corollary}
\newtheorem{conjecture}[theorem]{Conjecture}

\newtheorem{problem}[theorem]{Problem}

\newcommand{\eps}{\varepsilon}
\newcommand{\F}{\mathbb{F}}
\newcommand{\T}{\mathbb{T}}
\newcommand{\R}{\mathbb{R}}
\newcommand{\Z}{\mathbb{Z}}
\newcommand{\E}{\mathbb{E}}

\newcommand{\etal}{{et al.}}
\newcommand{\polylog}{\text{polylog}}
\newcommand{\poly}{\text{poly}}
\newcommand{\adeg}{\overline{\deg}}
\newcommand{\adegr}{\widetilde{\deg}}
\newcommand{\AC}{\text{AC}^0}
\newcommand{\ACC}{\text{ACC}^0}
\newcommand{\cP}{\mathrm{P}}

\newcommand{\NEXP}{\mathrm{NEXP}}
\newcommand{\NQP}{\mathrm{NQP}}

\newcommand{\modone}{\pmod{1}}

\title{Torus polynomials: an algebraic approach to ACC lower bounds}

\author{
Abhishek Bhrushundi\thanks{Part of this work was done when the author was visiting the University of California, San Diego. Research supported in part by Rutgers AAUP-AFT TA-GA Professional Development Fund, and by NSF grant CCF-1614023.}\\
Rutgers University\\
\texttt{abhishek.bhr@rutgers.edu}
\and
Kaave Hosseini\thanks{Supported by NSF grant CCF-1614023.}\\
University of California, San Diego\\
\texttt{skhossei@ucsd.edu}
\and
Shachar Lovett\thanks{Supported by NSF grant CCF-1614023.}\\
University of California, San Diego\\
\texttt{slovett@ucsd.edu}
\and
Sankeerth Rao\thanks{Supported by NSF grant CCF-1614023.}\\
University of California, San Diego\\
\texttt{skaringu@ucsd.edu}
}

\begin{document}
\maketitle

\abstract{
We propose an algebraic approach to proving circuit lower bounds for $\ACC$ by defining and studying the notion of \textit{torus polynomials}. We show how currently known polynomial-based approximation results for $\AC$ and $\ACC$ can be reformulated in this framework, implying that $\ACC$ can be approximated by low-degree torus polynomials. Furthermore, as a step towards proving $\ACC$ lower bounds for the majority function via our approach, we show that MAJORITY cannot be approximated by low-degree \textit{symmetric} torus polynomials. We also pose several open problems related to our framework.
}
\section{Introduction}

A major goal of complexity theory is to prove Boolean circuit lower bounds, i.e. find explicit Boolean functions that cannot be computed by small size circuits of a given type.
Over the years, three general approaches have been developed to achieve this.

The first approach is based on random restrictions. It applies to circuit classes in which
functions simplify when most inputs are fixed to random values. Classic examples
are the proofs by H{\aa}stad that $\AC$, i.e. polynomial size circuit families
of constant depth consisting of AND, OR, and NOT gates, cannot compute or approximate the PARITY function \cite{haastad1987computational}, and the shrinkage of De Morgan formulas (Boolean circuits consisting of AND, OR, and NOT gates whose underlying graph is a tree) under random restrictions \cite{haastad1998shrinkage}.
However, random restrictions don't seem to be useful against more powerful circuit classes such as $\AC[\oplus]$ --- the class of $\AC$ circuits equipped with PARITY gates.

The second approach is based on approximation by low-degree polynomials. Razborov \cite{razborov1987lower} and Smolensky \cite{smolensky1987algebraic} used this approach to prove lower bounds for $\AC[\oplus]=\AC[2]$, and more generally for $\AC[p]$ for any prime $p$ (This is the class of $\AC$ circuits that are allowed to have MOD$_p$ gates \footnote{a MOD$_p$ gate outputs $1$ if and only if the sum of its inputs is congruent to a non-zero value modulo $p$.}). This technique is based on showing that any function in the circuit class can be approximated by a low-degree polynomial over the finite field $\F_p$. Then, functions that do not admit such an approximation are provably outside the circuit class. A classic example here is that the MAJORITY function cannot be approximated by a low-degree polynomial over $\F_p$, and thus cannot be computed by $\AC[p]$. However, this method also breaks down when considering more powerful circuit classes such as $\AC[6]$, and more generally $\ACC$, i.e. $\AC$ circuits with MOD$_m$ gates where $m$ is a composite that is not a prime power.

The third method involves designing nontrivial satisfiability algorithms and then using them along with classical tools from structural complexity theory (among other techniques and results) to prove circuit lower bounds against $\ACC$ for functions in high complexity classes such as $\NEXP$. Williams \cite{williams2014nonuniform} used this approach to prove that $\NEXP\not\subseteq \ACC$, and very recently, Williams and Murray \cite{murray2018circuit} have extended this to show that $\NQP\not\subseteq \ACC$.

The goal of this paper is to focus on the second approach, namely the use of algebraic techniques, and to try and extend these techniques to prove lower bounds against $\ACC$. We show that an extension of finite field polynomials, which we call \emph{torus polynomials}, is a concrete candidate to achieve this. In particular, using a slightly stronger version of a result of Green \etal~\cite{green1992power}, we show that functions in $\ACC$ can be approximated\footnote{The notion of approximation that we use will be made explicit in \Cref{sec:intro-torus}.} by low-degree torus polynomials. 
We remark that torus polynomials also generalize the class of \emph{nonclassical polynomials} which arose in number theory and in higher order Fourier analysis \cite{tao2012inverse}, and are closely related to them.

This new characterization of $\ACC$ using torus polynomials raises a host of questions on the approximation of Boolean functions by torus polynomials, the most remarkable being the problem of finding an explicit Boolean function that cannot be approximated by low-degree torus polynomials; an answer to this question would imply $\ACC$ lower bounds. In this paper, we take steps towards trying to resolve this question by initiating the study of approximation of Boolean functions by torus polynomials and proving some interesting results along the way. The motivation for our work is two-fold:
\begin{enumerate}
\item Given the slew of recent works exploring properties and applications of nonclassical polynomials\cite{terry2008blog,tao2012inverse,bhatta2013testing,bhowmick2015nonclassical,bhrushundi2017polynomial}, and the fact that torus polynomials are closely related to nonclassical polynomials, we believe that our characterization of $\ACC$ using torus polynomials might pave a way for new $\ACC$ lower bounds.
\item  While the works of Williams \cite{williams2014nonuniform} and Williams and Murray \cite{murray2018circuit} are groundbreaking and prove nontrivial lower bounds against $\ACC$, their proofs are not purely combinatorial/algebraic, and it will be interesting to recover their results using purely algebraic/combinatorial techniques. We hope that our work will renew interest in this line of inquiry.  
\end{enumerate}
\subsection{Torus polynomials}
\label{sec:intro-torus}

Let $\T=\R/\Z$ denote the one-dimensional torus. A \emph{torus polynomial} is simply a real polynomial restricted to the domain $\{0,1\}^n$ and evaluated modulo one\footnote{For $x \in \mathbb{R}$, $x$ modulo one, denoted by $x$ mod $1$, is equal to the fractional part of $x$ given by $x - \lfloor x \rfloor$, where $\lfloor x \rfloor$ is the floor function. For example, $2.6$ mod $1$ is $0.6$, and $-1.3$ mod $1$ is $0.7$.}. Namely,
a degree-$d$ torus polynomial $P:\{0,1\}^n \to \T$ is
$$
P(x) = \sum_{S \subseteq [n], |S| \le d} P_S \prod_{i \in S} x_i \modone,
$$
where $P_S \in \R$. 

As it shall become evident later, torus polynomials extend finite field polynomials in that they provide a uniform way to capture computation of Boolean functions by polynomials over different finite fields --- if a function can be computed by a low-degree polynomial over a finite field then it can be approximated by a low-degree torus polynomial. We will discuss this in detail in \Cref{sec:approx}.

For $z \in \T$,  let $\iota(z)$ denote the unique representative of $z$ in $[-1/2,1/2)$ (e.g., $\iota(0.4) = 0.4$ and $\iota(0.7) = -0.3$). Then we can define its norm, denoted by $|z \modone|$, to be
$$|z \modone| = |\iota(z)|.$$
For $F:\{0,1\}^n \to \T$, define
$$\|F \modone\|_{\infty} := \max_{x \in \{0,1\}^n} |F(x) \modone|.$$ We embed Boolean functions as functions mapping into the torus by enforcing their output to be in $\{0,1/2\} \subset \mathbb{T}$ (This can be achieved by scaling the output of the function by $1/2$). The following is the main definition of approximation that we consider:

\begin{definition}
Let $f:\{0,1\}^n \to \{0,1\}$ be a Boolean function. For $\eps>0$, a torus polynomial $P:\{0,1\}^n \to \T$ is said to $\epsilon$-approximate $f$ if
$$
\left\|P  - \frac{f}{2} \modone \right\|_{\infty} \le \eps.
$$
\end{definition}
Intuitively, a torus polynomial that approximates $f$ takes a value ``close'' to $0$ in the torus $\mathbb{T}$ whenever $f$ takes the value $0$, and takes a value ``close'' to $1/2$ in the torus whenever $f$ takes the value $1$.  

We now introduce the notion of the \textit{toroidal approximation degree} of a Boolean function.

\begin{definition}[Toroidal approximation degree of Boolean functions]
Let $f:\{0,1\}^n \to \{0,1\}$ be a Boolean function. For $\eps>0$, the toroidal $\eps$-approximation degree of $f$ is the minimal $d \ge 0$,
for which there exists a torus polynomial $P:\{0,1\}^n \to \T$ of degree $d$, that satisfies
$$
\left\|P  - \frac{f}{2} \modone \right\|_{\infty} \le \eps.
$$
We denote this by $\adeg_{\eps}(f)=d$.
\end{definition}

We illustrate in \Cref{sec:approx}, in increasing generality, the power of torus polynomials.
The most general result
(\Cref{cor:acc}) shows that if $f$ can be computed by an $\ACC$ circuit then
$$
\adeg_{\eps}(f) \le \polylog(n/\eps).
$$
The proof of this result uses a slightly stronger version of a result of Green \etal~\cite{green1992power}. 

The above characterization paves way for a new approach to proving lower bounds against $\ACC$ for an explicit function, ideally in the class $\cP$. Concretely, we pose the following open problem.

\begin{problem}
\label{problem:explicit}
Find an explicit function $f:\{0,1\}^n \to \{0,1\}$ in $\cP$ whose toroidal
$\eps$-approximation degree is $\omega(\polylog(n/\eps))$. By \Cref{cor:acc},
it cannot be computed by $\ACC$ circuits.
\end{problem}
Williams \cite{williams2014nonuniform} proved that $\NEXP\not\subseteq \ACC$ via designing nontrivial satisfiability algorithms for $\ACC$, 
and Williams and Murray \cite{murray2018circuit} improved the approach to show that $\NQP\not\subseteq \ACC$. Thus, an intermediate goal towards resolving \Cref{problem:explicit} is to prove toroidal approximation lower bounds for functions $f \in \NEXP$ or $f \in \NQP$.

A long-standing open problem in circuit complexity is to show that MAJORITY cannot be computed in $\ACC$. Thus
the following question is natural.
\begin{problem}
What is the toroidal $\eps$-approximation degree of MAJORITY?
\end{problem}
How can one go about answering this question? We now turn to the setting of approximation of Boolean functions by real polynomials -- which prima facie shares some similarities with our setting -- for inspiration, highlighting the main differences between the two notions.

\subsection{Comparison with real polynomials}
\label{sec:intro-compare-real}

Given a function $f:\{0,1\}^n\rightarrow \{0,1\}$, the real $\eps$-approximation degree of $f$, denoted by $\adegr_\eps(f)$, is the minimal $d$ such that there is a real polynomial $P$ of degree $d$ such that $\|f-P\|_\infty \leq \eps$ (this is the $\ell_\infty$-norm restricted to the domain $\{0,1\}^n$). It is clear that $\adeg_\eps(f) \le \adegr_\eps(f)$.

A beautiful result of Nisan and Szegedy \cite{nisan1992degree} shows that the real $\eps$-approximation degree of MAJORITY is $\Omega(\sqrt{n})$ for $\eps < 1/2$. Their proof proceeds in two stages: (i) showing that if a \textit{symmetric} real polynomial $\eps$-approximates MAJORITY then it must have degree $\Omega(\sqrt{n})$; and (ii) that any polynomial that $\eps$-approximates MAJORITY can be \textit{symmetrized} and made into a symmetric polynomial with the same degree and approximation guarantee. 

Attempting to follow the same strategy in the case of torus polynomials, we show in \Cref{cor:maj} in \Cref{sec:lower} that if
one restricts attention to \textit{symmetric} torus polynomials (namely, symmetric real polynomials evaluated modulo one), then the toroidal $(1/20n)$-approximation degree of MAJORITY is $\Omega(\sqrt{n/\log n})$.

Unfortunately, the aforementioned idea of symmetrization cannot be used in the setting of torus polynomials in a straightforward manner and so it's unclear how powerful non-symmetric torus polynomials are compared to their symmetric counterparts. We conjecture that they are not any better at approximating MAJORITY than symmetric torus polynomials:
\begin{conjecture}
The toroidal $(1/20n)$-approximation degree of MAJORITY is $\Omega(\sqrt{n/\log n})$.
\end{conjecture}
We remark that a positive answer to the above conjecture will give an algebraic proof that MAJORITY is not in $\ACC$.

Let $\Delta_w:\{0,1\}^n \to \{0,1\}$ denote the delta function which takes the value $1$ on inputs of Hamming weight $w$ and is $0$ elsewhere. En route to proving the aforementioned lower bound for MAJORITY we also prove lower bounds for the delta functions in \Cref{thm:delta-lb}, showing that one needs symmetric torus polynomials of degree $\Omega(\sqrt{n/\log n})$ in order to be able to $(1/20n)$-approximate the delta functions. \\
Somewhat surprisingly, for relatively large values of $\eps$, the delta functions can be nontrivially $\eps$-approximated by low-degree \textit{symmetric} torus polynomials. In particular, we show in \Cref{lemma:delta} in \Cref{sec:upper} that for every delta function there is a symmetric torus polynomial of degree $\polylog(n/\eps)/\eps$ that $\eps$-approximates it, and thus
$$
\adeg_{\eps}(\Delta_w) \le \frac{\polylog(n/\eps)}{\eps}.
$$
This kind of dependence of the toroidal approximation degree on $\eps$ is quite interesting, and is unlike the case of real approximation --- the real approximation degree of the delta functions is $\Omega(\sqrt{n})$ for both small and large values of $\eps$. In fact, for constant $\eps$, this also shows a super-polynomial separation between real and toroidal approximation degree.

 This also highlights other major differences between the real and the toroidal setting. Nisan and Szegedy \cite{nisan1992degree} show that for every Boolean function the real approximation degree is polynomially related to the degree of exact representation by real polynomials. However, in the case of torus polynomials, this is not true: the delta functions require the degree to be $\Omega(n)$\footnote{To see this, note that the delta function $\Delta_n(x)$ has a unique representation as a torus polynomial given by $\Delta_n(x) = \frac{x_1\cdots x_n}{2}$.} for exact representation whereas their toroidal $1/3$-approximation degree is $O(\polylog(n))$.

An interesting property of real approximation is its amenability to amplification, namely the fact that, for any Boolean function $f$ and $\eps < 1/3$, given a polynomial $p$ of degree $d$ that $1/3$-approximates $f$, it can be transformed into a polynomial $p'$ of degree $d' = O(d \log(1/\eps))$ that $\eps$-approximates $f$. In other words, $\adegr_{\eps}(f) \le O(\adegr_{1/3}(f) \log(1/\eps))$.
It is not clear whether such a transformation is possible in the case of toroidal approximation. In the case of real approximation, the transformation is symmetry preserving, but, given the results for the delta functions discussed in the previous paragraphs, we should not expect this in the toroidal case. This motivates the following problem.
\begin{problem}
How is $\adeg_\eps(f)$ related to $\adeg_{1/3}(f)$?
\end{problem}

\subsection{Comparison with nonclassical polynomials}
\label{sec:intro-nonclassical}
As mentioned before, torus polynomials generalize the class of nonclassical polynomials (this will be evident from the definition of nonclassical polynomials stated below). We remark that the results of this paper can be similarly phrased in terms of nonclassical polynomials instead of torus polynomials.
This is because for the purpose of approximation of Boolean functions -- which is the topic of this paper -- torus polynomials and nonclassical polynomials are equivalent, as we shall see below. However, torus polynomials are simpler to describe (they are just real polynomials evaluated modulo $1$) and more elegant (they are field independent),
and hence we believe are a better choice for an algebraic model and for stating our results.

We now give the definition of nonclassical polynomials; here we provide what is known as the global definition of nonclassical polynomials over $\{0,1\}^n$.
For simplicity, we restrict our attention to nonclassical polynomials defined over $\F_2^n$, but note that the results generalize to nonclassical polynomials defined over $\F_p^n$ for any constant prime $p$.

\begin{definition}[Nonclassical polynomials]
	A function $Q:\{0,1\}^n \rightarrow \T$ is a nonclassical polynomial (over $\F_2$) of degree at most $d$ if and only if it can be written as
	$$Q(x) = \alpha + \sum_{\emptyset \subset S\subseteq [n]; k\geq 0;\\ 0 <|S|+k \leq d} \frac{c_{S,k}}{2^{k+1}} \prod_{i\in S} x_i \modone$$ where $c_{S,k}\in \{0,1\}$ and $\alpha \in \T$.
\end{definition}
The following simple claim shows that torus polynomials can be approximated by nonclassical polynomials.
\begin{claim}
\label{claim:nonclassical}
Let $P:\{0,1\}^n \rightarrow \T$ be a torus polynomial of degree at most $d$ and let $\eps\in (0,1)$. Then there exists a nonclassical polynomial $Q$ of degree at most $O(d\log n + \log(1/\eps))$ such that $\|P-Q \modone\|_\infty \leq \eps$.
\end{claim}
\begin{proof}
	Suppose $P(x) = \alpha+ \sum_{\emptyset \subset S \subseteq [n], |S| \le d} P_S \prod_{i \in S} x_i \modone$. We can assume without loss of generality that $P_S \in [0,1)$ for all $S$. We approximate each $P_S$ separately using dyadic rationals. Let $P_S = 0.c_{S,0} c_{S,1} c_{S,2} \ldots$, where $c_{S,i} \in \{0,1\}$, be
its binary expansion. Let $t \ge 1$ be a parameter that we will fix later, and note that $$\left|P_S - \sum_{0\leq k\leq t} \frac{c_{S,k}}{2^{k+1}}\right|\leq 2^{-t}.$$
Define the nonclassical polynomial
$$Q(x) = \alpha + \sum_{\emptyset \subset S\subseteq [n];  k\geq 0;\\ 0 <|S|+k \leq t+d} \frac{c'_{S,k}}{2^{k+1}} \prod_{i\in S} x_i \modone,$$
where $c'_{S,k} = c_{S,k}$ for $|S| \le d,k \le t$, and is $0$ otherwise. Then $\deg(Q)\le t + d$, and $$|P(x)-Q(x) \modone| \le {n \choose \le d} 2^{-t}$$ for all $x \in \{0,1\}^n$.
Setting $t=O(d \log n + \log(1/\eps))$ completes the proof.
\end{proof}

Recall that our goal, motivated by proving $\ACC$ lower bounds, is to find a Boolean function which cannot be $1/\poly(n)$-approximated by a torus polynomial of degree $\polylog(n)$.
Given \Cref{claim:nonclassical}, this is equivalent to the problem of finding a Boolean function which cannot be $1/\poly(n)$-approximated by a nonclassical polynomial of degree $\polylog(n)$.
As we mentioned before, owing to the elegance and ease of description of torus polynomials relative to nonclassical polynomials, torus polynomials make for a more convenient choice in our setting.

\subsection{Comparison with other notions of approximation}
\label{subsec:compare}
It's clear from our discussion in the previous section that torus polynomials are closely related to nonclassical polynomials, and so it's worthwhile to discuss two notions of approximation of Boolean functions by nonclassical polynomials that have been studied in the literature. The first deals with the \textit{exact} computation of a Boolean function by a nonclassical polynomial on a nontrivial fraction of the domain~\cite{bhowmick2015nonclassical}. For example, the work of Bhrushundi \etal \cite{bhrushundi2017polynomial} shows that any polynomial that computes MAJORITY correctly even on two-thirds of the points must have degree $\Omega(\sqrt{n})$. While many of these bounds for nonclassical polynomials should also hold for torus polynomials, we remark that they are not relevant to our setting since our notion of approximation (i.e., point-wise) is incomparable to the above notion.

The second notion is that of correlation with polynomials, which was studied, for example, by Bhowmick and Lovett\cite{bhowmick2015nonclassical}. Without getting into definitions here, we note that this notion of approximation is \textit{weaker} than that of point-wise approximation\footnote{By this we mean that if a function is point-wise approximated by a low-degree torus polynomial then it is also approximated by that polynomial in the correlation sense.}, and thus for the purpose of proving lower bounds for $\ACC$ it makes sense to work with only the latter. This also means that the upper bound results proved in the work of Bhowmick and Lovett (i.e., showing how certain Boolean functions can be approximated by low-degree nonclassical polynomials in the correlation sense) don't have any implications for our setting. Even their lower bound results, unfortunately, are not useful for us given that they only work for polynomials of degree $<< \log(n)$, whereas we are dealing with polynomials of degree $\polylog(n)$. 
\subsection{Natural proofs}
\label{sec-natural}

The natural proofs barrier of Razborov and Rudich \cite{razborov1997naturalproof} isn't really a problem for our approach since we are only trying to prove lower bounds against $\ACC$ and pseudorandom generators are not believed to be contained in this class. It is also not clear whether the property in question, i.e. (in)approximability by torus polynomials, is \textit{natural}, and, in particular, it will be interesting to investigate whether one can efficiently distinguish between Boolean functions which can be approximated by low-degree torus polynomials and a random Boolean function, i.e. whether this property is \textit{constructive}:

\begin{problem}
Given the truth table of a function $f:\{0,1\}^n\to \{0,1\}$ and $\eps>0$, decide in polynomial time (in $2^n$ and $1/\eps$)
whether $\adeg_{\eps}(f) \leq \polylog(n/\eps)$.
\end{problem}

\paragraph{Paper organization.}
In \Cref{sec:approx}, we prove toroidal approximation results for Boolean functions in bounded circuit classes such as $\AC[p]$ and $\ACC$.
In \Cref{sec:lower}, we prove lower bounds against symmetric torus polynomials approximating the MAJORITY function and the delta functions.
In \Cref{sec:upper}, we show that symmetric torus polynomials have surprising power in approximating the delta functions when the error $\eps$ is not too small. We introduce definitions and notation along the way, as and when needed.

\section{Approximation of circuit classes}
\label{sec:approx}
In this section, we illustrate how the framework of approximation by torus polynomials captures computation of Boolean functions in various models of computation. We begin by showing that functions that are computable by low-degree polynomials over finite fields can be approximated by low-degree torus polynomials. \\
It might be instructive to keep in mind that, for the scope of the entire paper, whenever we consider polynomials we restrict ourselves to only \textit{multilinear} polynomials, i.e. polynomials in which the maximum degree of any variable is at most $1$. Even if we encounter polynomials that do not adhere to this form during intermediate steps in certain proofs, we can always \textit{multilinearize} the polynomials by making the degrees of all the variables equal to $1$ wherever they appear. It suffices to consider multilinear polynomials because we always restrict the variables to the domain $\{0,1\}$.
\subsection{Polynomials over finite fields}
Let $\F_p$ be a prime finite field. We say a polynomial $P(x) \in \F_p[x_1, \ldots, x_n]$ computes a Boolean function $f$ if $$\forall x \in \{0,1\}^n,\ f(x) = P(x).$$
 Consider a function $f$ which is computed by a low-degree polynomial over $\F_p$. 
We will now show that it can be approximated by a low-degree torus polynomial. We would require the following theorem on modulus-amplifying
polynomials of Beigel and Tarui~\cite{beigel1991acc}, following previous results of Toda~\cite{toda1991pp} and Yao~\cite{yao1985separating}.

\begin{lemma}[Beigel and Tarui \cite{beigel1991acc}]
\label{lemma:modamp}
For every $k \ge 1$, there exists a univariate polynomial $A_k:\Z \to \Z$ of degree $2k-1$ such that the following holds. For every $m \ge 2$,
\begin{itemize}
\item If $x \in \Z$ satisfies $x \equiv 0 \pmod{m}$ then $A_k(x) \equiv 0 \pmod{m^k}$.
\item If $x \in \Z$ satisfies $x \equiv 1 \pmod{m}$ then $A_k(x) \equiv 1 \pmod{m^k}$.
\end{itemize}
\end{lemma}

\begin{lemma}
\label{lemma:poly_Fp}
Let $f:\{0,1\}^n \to \{0,1\}$. Assume that $f$ can be computed by a polynomial over $\F_p$ of degree $d$.
Then for every $\eps>0$,
$$
\adeg_{\eps}(f) \le O(d \log(1/\eps)).
$$
\end{lemma}

\begin{proof}
Since $f$ is computable by degree-$d$ polynomials over $\F_p$, there must be an integer polynomial $F(x)$ (i.e., a polynomial with coefficients in $\mathbb{Z}$) of degree $d$ such that
$$
F(x) \equiv f(x) \pmod{p} \qquad \forall x \in \{0,1\}^n.
$$
Let $k \ge 1$ be large enough so that $1/p^k \le \eps$. Let $0 \le q \le p^k-1$ be such that $$\left| \frac{q}{p^k} - \frac{1}{2} \modone\right| \le \eps.$$
Define
$$
G(x) = \frac{q A_k(F(x))}{p^k} \modone.
$$
We claim that \begin{align}\label{eq:bound1}\left|G(x)-\frac{f(x)}{2}\modone\right| \le \eps\end{align} for all $x$. To see this, fix $x$, and recall that $F(x) \equiv f(x) \pmod{p}$, which means
that $A_k(F(x)) \equiv f(x) \pmod{p^k}$, and hence $G(x) \equiv \frac{q}{p^k} f(x) \modone$. \eqref{eq:bound1} now follows from our choice of $q$.

Noting that the degree of $G$
is $(2k-1) d \le O(d \log(1/\eps))$ completes the proof.
\end{proof}

We will later need the following simple variant of \Cref{lemma:poly_Fp}. Its proof is identical.

\begin{lemma}
\label{lemma:poly_Fp_variant}
Let $f:\{0,1\}^n \to \{0,1\}$. Assume that $f$ can be computed by a polynomial over $\F_p$ of degree $d$.
Then for every $\alpha \in [0,1]$ and every $\eps>0$, there exists a torus polynomial $P:\{0,1\}^n \to \T$ of degree $O(d \log(1/\eps))$ such that
$$
\|P - \alpha f \modone\|_\infty \le \eps.
$$
\end{lemma}

\subsection{Circuit class $\AC[p]$}
Recall that, for a fixed prime $p$, $\AC[p]$ is the class of functions computable by polynomial size circuits of constant depth, consisting of AND, OR, NOT, and MOD$_p$ gates. Here a MOD$_p$ gate is one that outputs $1$ if and only if the sum of its inputs is congruent to a non-zero value modulo $p$.
  
Let $f:\{0,1\}^n \to \{0,1\}$ be a function in $\AC[p]$. We show that it can also be approximated by low-degree torus polynomials.
The starting point is the classic result of Razborov~\cite{razborov1987lower} and Smolensky~\cite{smolensky1987algebraic} which
shows that $\AC[p]$ circuits can be approximated by random low-degree polynomials over $\F_p$ in the following sense.

\begin{theorem}[Razborov-Smolensky\cite{razborov1987lower,smolensky1987algebraic}]
\label{thm:RS}
Let $f:\{0,1\}^n \to \{0,1\}$ be computed by an $\AC[p]$ circuit. Then for every $\eps>0$,
there exists a distribution $\nu$ supported on polynomials $F:\F_p^n \to \{0,1\}$
of degree $d=\polylog(n/\eps)$ such that
$$
\Pr_{P \sim \nu}[P(x) = f(x)] \ge 1-\eps \qquad \forall x \in \{0,1\}^n.
$$
\end{theorem}
We can assume without loss of generality that all the polynomials in the support of the distribution $\nu$ have range $\{0,1\}$. This is because given an arbitrary polynomial $P(x)$ over $\F_p$ we can convert it into the polynomial $P'(x) = (P(x))^{p-1}$ which has range $\{0,1\}$ by Fermat's little theorem. Note that the degree of $P'$ is at most $p$ times the degree of $P$ which is not really a problem since $p = O(1)$ for us.

We now show why torus polynomials approximate $\AC[p]$ functions.
\begin{lemma}
\label{lemma:compose_distrib}
Let $f:\{0,1\}^n \to \{0,1\}$. Assume that there exists
a distribution $\nu$ supported on polynomials $F:\F_p^n \to \{0,1\}$
of degree $d$ such that
$$
\Pr_{P \sim \nu}[P(x) = f(x)] \ge 1-\eps \qquad \forall x \in \{0,1\}^n.
$$
Then
$$
\adeg_{3 \eps}(f) \le O(d \log(n/\eps)).
$$
\end{lemma}

\begin{proof}
By standard Chernoff bounds, if we sample $F_1,\ldots,F_m \sim \nu$ independently for $m=O(n/\eps^2)$ then with high probability,
$$
|\{i \in [m]: F_i(x) \ne f(x)\}| \le 2 \eps m \qquad \forall x \in \{0,1\}^n.
$$
Fix such a sample. Recall that $F_i:\F_p^n \to \{0,1\}$ are computed by degree $d$ polynomials over $\F_p$.
Next, apply \Cref{lemma:poly_Fp_variant} with $\alpha=1/2m$ and error $\eps/m$.
This gives us torus polynomials $P_i:\{0,1\}^n \to \T$ of degree $O(d \log(m/\eps))$ such that
$$
\left|P_i(x) - \frac{1}{2m} F_i(x) \modone \right| \le \frac{\eps}{m} \qquad \forall x \in \{0,1\}^n.
$$
Finally, take
$$
P(x) = P_1(x) + \ldots + P_m(x) \modone.
$$
We claim that $P(x)$ is a torus polynomial which $3 \eps$-approximates $f(x)$. To see this, fix $x \in \{0,1\}^n$, and observe that
$$
\left| P(x) - \frac{F_1(x)+\ldots+F_m(x)}{2m} \modone \right| \le \eps
$$
and
$$
\left| \frac{F_1(x)+\ldots+F_m(x)}{2m} - \frac{f(x)}{2} \modone \right| \le 2 \eps,
$$
and so
$$\left|P(x) - \frac{f(x)}{2} \modone\right| \le 3\eps.$$
This means that$$\adeg_{3 \eps}(f) \le \deg(P) = \max\{\deg(P_i): i \in [m]\} = O(d \log(m/\eps)) = O(d \log(n/\eps)).
$$
\end{proof}

\begin{corollary}
Let $f:\{0,1\}^n \to \{0,1\}$ be a function in $\AC[p]$. Then for every $\eps>0$,
$$
\adeg_{\eps}(f) \le \polylog(n/\eps).
$$
\end{corollary}

An interesting question that is motivated by the above results is whether we can have a mini-max type theorem for torus polynomials. Lemma~\ref{lemma:compose_distrib}
gives such a theorem in a very limited regime. The following is an attempt to generalize this.

\begin{problem}
Let $f:\{0,1\}^n \to \{0,1\}$. Assume that for any distribution $\nu$ over $\{0,1\}^n$, there exists a low-degree torus polynomial $P_{\nu}:\{0,1\}^n \to \T$ such that
$$
\E_{x \sim \nu}\left[ \left| P_{\nu}(x) - \frac{f(x)}{2} \modone \right|\right] \le \eps.
$$
Does that imply that the toroidal approximation degree of $f$ is small? That is, does there exist a single low-degree torus polynomial which approximates $f$ on all inputs?
\end{problem}

It might also be useful to assume the stronger assumption that for
any distribution $\nu$ over $\{0,1\}^n$ and any $\alpha \in [0,1]$
there exists a torus polynomial $P_{\nu,\alpha}:\{0,1\}^n \to \T$ of degree $d$ such that
$$
\E_{x \sim \nu} \left[\left| P_{\nu,\alpha}(x) - \alpha f(x) \modone \right|\right]\le \eps.
$$
This is also related to the following problem.

\begin{problem}
Let $f:\{0,1\}^n \to \{0,1\}$. For any $\alpha \in [0,1]$ and $\eps>0$ define $d(\alpha,\eps)$ to be the minimal degree of a torus polynomial $P:\{0,1\}^n \to \T$ such that
$$
\|P - \alpha f \modone \|_{\infty} \le \eps.
$$
What is the behavior of $d(\alpha,\eps)$ as a function of $\alpha$ and of $\eps$? Specifically,
\begin{itemize}
\item Can we bound $\max_{\alpha} d(\alpha,\eps)$ in terms of $d(1/2,\eps)$?
\item Can we bound $\max_{\alpha} d(\alpha,\eps)$ in terms of $\max_{\alpha} d(\alpha,0.1)$?
\end{itemize}
\end{problem}

\subsection{Circuit class $\ACC$}
We now turn our attention to $\ACC$ functions and show that they too can be approximated by low-degree torus polynomials. Recall that a function is in $\ACC$ if it can be computed by polynomial size circuits of constant depth with AND, OR, NOT, and MOD$_m$ gates where $m$ may be composite.

Our starting point is the following result of Green \etal~\cite{green1992power} which extends previous results of~\cite{yao1985separating,beigel1991acc}.

\begin{theorem}[Green et al. \cite{green1992power}]
\label{thm:green_old}
Let $f:\{0,1\}^n \to \{0,1\}$ be computable by $\ACC$ circuits of depth $\ell$ and size $\poly(n)$. Then for any $e \ge 1$ there exists an integer polynomial $F(x)$ of
degree $d=e^{O(\ell)}\log^{O(\ell^2)}n$ which satisfies the following: there is some $k  \ge e$ such that 
\begin{center}
$\forall x \in \{0,1\}^n, \ $ $F(x) = f(x)2^k + E(x) \pmod{2^{k+e}}$ 
\end{center}
for some error $E(x) \leq 2^{k-1}$.
\end{theorem}

Note that the above theorem states that the $k$\textsuperscript{th} bit of $F(x)$ in binary always equals to $f(x)$ and that it's padded with $e-1$ zeros to its left, i.e the $(k+1)$\textsuperscript{th}, $(k+2)$\textsuperscript{th}, \dots, $(k+e-1)$\textsuperscript{th} bits are all guaranteed to be equal to $0$. It turns out that, implicit in their work, is the following slightly stronger version of the above result which lets us pad zeros on both sides of the output bit (i.e., the $k$\textsuperscript{th} bit).

\begin{theorem}[Implicit in Green et al. \cite{green1992power}]
\label{thm:green}
Let $f:\{0,1\}^n \to \{0,1\}$ be computable by $\ACC$ circuits of depth $\ell$ and size $\poly(n)$. Then for any $e \ge 1$ there exists an integer polynomial $F(x)$ of
degree $d=e^{O(\ell)}\log^{O(\ell^2)}n$ which satisfies the following: there is some $k  \ge e$ such that 
\begin{center}
$\forall x \in \{0,1\}^n, \ $ $F(x) = f(x)2^k + E(x) \pmod{2^{k+e}}$ 
\end{center}
for some error $E(x) \leq 2^{k-e}$.
\end{theorem}
Note the difference between the statements of \Cref{thm:green_old} and \Cref{thm:green}: while the former upper-bounds the error $E(x)$ by $2^{k-1}$ the latter bounds it by $2^{k-e}$, thus padding the output bit with $e-1$ zeros on both the sides. 

We remark that the proof of \Cref{thm:green} is essentially the same as that of \Cref{thm:green_old}, with some minor tweaks, and so we omit it here. We now show how to use \Cref{thm:green} to prove that low-degree torus polynomials approximate functions in $\ACC$.

\begin{corollary}
\label{cor:acc}
Let $f:\{0,1\}^n \to \{0,1\}$ be a function in $\ACC$. Then for every $\eps>0$, there is a torus polynomial of degree $\polylog(n/\eps)$ that $\eps$-approximates $f$. In other words,
$$
\adeg_{\eps}(f) \le \polylog(n/\eps).
$$
\end{corollary}

\begin{proof}
Let us assume that $f$ is computable by $\ACC$ circuits of size $\poly(n)$ and depth $\ell$. Recall that, by definition of $\ACC$, $\ell = O(1)$. Let $F(x)$ be the polynomial obtained by applying \Cref{thm:green} to $f$ with $e=\log(1/\eps)$ such that for some $k \ge e$
$$\forall x \in \{0,1\}^n, \  F(x) = f(x)2^k + E(x) \pmod{2^{k+e}}.$$
The degree of $F(x)$ is $d=e^{O(\ell)}\log^{O(\ell^2)}n = \polylog(n/\eps)$. Define the following torus polynomial
$$
P(x) = \frac{F(x)}{2^{k+1}} \modone.
$$
Clearly $\deg(P)=d$. For $i \ge 0$, let $F_i(x)$ denote the $i$\textsuperscript{th} bit of $F(x)$. Then, by the definition of $F$,
$$
\frac{F(x)}{2^{k+1}} \modone = \sum_{i=0}^{k} 2^{i-k-1} F_i(x) \modone = \frac{f(x)}{2} + \sum_{i=0}^{k-e} 2^{i-k-1} F_i(x) \modone.
$$
As $F_i(x) \in \{0,1\}$ for all $i$, we can bound
$$
\left| P(x) - \frac{f(x)}{2} \modone \right| \le 2^{-e} \le \eps \qquad \forall x \in \{0,1\}^n.
$$
\end{proof}

\section{Lower bound for symmetric torus polynomials}
\label{sec:lower}

In this section we prove a lower bound on the degree of \emph{symmetric} torus polynomials that approximate MAJORITY. It will be instructive to think of symmetric torus polynomials as symmetric real polynomials evaluated modulo one. We start by examining the question for delta functions.

For $x \in \{0,1\}^n$, let $|x|=\sum x_i$ denote its Hamming weight.
The delta function $$\Delta_w:\{0,1\}^n \to \{0,1\},$$ for $0 \le w \le n$, is defined as
$$
\Delta_w(x) =
\begin{cases}
1 & |x|=w\\
0 & \text{otherwise}
\end{cases}\;.
$$
\begin{lemma}
\label{thm:delta-lb}
Let $n,d$ be positive integers such that for every $0 \le w \le n$ there exists a symmetric torus polynomial $Q_w:\{0,1\}^n \to \T$ of degree $d$ that $\frac{1}{20n}$-approximates $\Delta_w(x)$. Then $d = \Omega\left(\sqrt{\frac{ n}{\log n}}\right)$.
\end{lemma}
\begin{proof}
Let $\mathrm{Sym}(n)$ denote the set of symmetric Boolean functions in $n$ variables and let $\mathrm{SymPoly}_{d,k}(n)$ denote the set of symmetric torus polynomials in $n$ variables of degree $d$ whose coefficients are of the form $q/2^k$ for $q \in \{-(2^k -1), \ldots, 0, \ldots, 2^k -1\}$.

Let $f$ be an arbitrary function in $\mathrm{Sym}(n)$. Abusing notation, we let $f^{-1}(1)$ denote the set of weights of the layers of the Hamming cube where $f$ takes value $1$. Now define the torus polynomial $Q_f$ as
$$Q_f(x) = \sum_{i \in f^{-1}(1)} Q_i(x) \modone.$$
It follows that $Q_f$ is a symmetric torus polynomial of degree $d$ that $\frac{1}{20}$-approximates $f$. Since $Q_f$ is a symmetric torus polynomial, namely a symmetric real polynomial modulo one, it may be written without loss of generality as
$$Q_f(x) = \sum_{j = 0}^d c_j \left(\sum x_i\right)^j \modone,$$
where $c_j \in [0,1)$. Let $k \ge 0$ be an integer whose value we will fix later. For $0 \le j \le d$, let $q_j \in \{-(2^k -1 ), \ldots, 0, \ldots, 2^k-1\}$ be such that
$$\left| \frac{q_j}{2^k} - c_j \right| \le \frac{1}{2^k},$$
and define $Q'_f$ to be the polynomial
$$Q'_f(x) = \sum_{j = 0}^d \frac{q_j}{2^k} \cdot \left(\sum x_i\right)^j \modone.$$
Observe that for every $x \in \{0,1\}^n$,
$$\left|Q_f(x) - Q'_f(x) \modone\right| \le \sum_{j = 0}^d \left| \frac{q_j}{2^k} - c_j \right|\cdot {|x|}^j \le \frac{(d+1)\cdot n^d}{2^k}.$$
If $k$ is such that $\frac{(d+1)\cdot n^d}{2^k} \le \frac{1}{20}$ then
$$\left\|Q_f - Q'_f \modone \right\|_{\infty} \le \frac{1}{20}, $$
and so
$$\left\|\frac{f}{2} - Q'_f \modone \right\|_{\infty} \le \left\|\frac{f}{2} - Q_f \modone\right\|_{\infty} + \left\|Q_f - Q'_f \modone\right\|_{\infty} \le  \frac{1}{10}.$$
Note that we can choose $k=O(d \log n)$ while still satisfying the required condition on $k$.

So far we have shown that for every $f \in \mathrm{Sym}(n)$ there is a polynomial $Q_f \in \mathrm{SymPoly}_{d,k}(n)$ that $1/10$-approximates $f$ where $k = O(d \log n)$. In the other direction, one can easily verify that every polynomial in $\mathrm{SymPoly}_{d,k}(n)$ can $1/10$-approximate \textit{at most} one function in $\mathrm{Sym}(n)$. This implies that
$$ |\mathrm{SymPoly}_{d,k}(n)| \ge |\mathrm{Sym}(n)|.$$
Plugging in $|\mathrm{SymPoly}_{d,k}(n)| = 2^{(k+1)(d+1)}$ and $|\mathrm{Sym}(n)| = 2^n$, and using $k = O(d \log n)$, yields the bound $d = \Omega\left(\sqrt{\frac{n}{\log n}}\right)$.
\end{proof}
Before we proceed, we formally define MAJORITY on $n$ bits, denoted by $\mathrm{Maj}_n(x)$, as
$$
\mathrm{Maj}_n(x) =
\begin{cases}
1 & |x| \ge \frac{n}{2}\\
0 & \text{otherwise}
\end{cases}\;.
$$
\begin{lemma}
\label{lem:maj}
If there is a symmetric torus polynomial of degree $o\left( \sqrt{\frac{ n}{\log n}}\right)$ that $\frac{1}{20n}$-approximates $\mathrm{Maj}_n(x)$, then for every $0 \le w \le n$ there is a symmetric torus polynomial of degree $o\left( \sqrt{\frac{ n}{\log n}}\right)$ that $\frac{1}{20 n}$-approximates $\Delta_w(x)$.
\end{lemma}
\begin{proof}
Fix $w$. Let $\Delta_{\ge w}(x)$ denote the function that takes value $1$ iff $|x| \ge w$. Then we can write
\begin{equation}
\label{eq:delta-maj}
\Delta_{\ge w}(x_1, \ldots, x_n) = \mathrm{Maj}_{2n+1}(x_1, \ldots, x_n, c_1, \ldots c_{n+1}),
\end{equation}
where $c \in \{0,1\}^{n+1}$ is the string whose first $n - w +1$ bits are set to $1$ and the rest of the bits are set to $0$. Let $Q(x_1, \ldots x_{2n+1})$ be the symmetric torus polynomial in $2n+1$ variables that $\frac{1}{20(2n+1)}$-approximates $\mathrm{Maj}_{2n+1}(x)$. Let $Q_{\ge w}(x_1, \ldots, x_n)$ be the torus polynomial defined as
$$Q_{\ge w}(x_1, \ldots x_n) = Q(x_1, \ldots, x_n, c_1, \ldots, c_{n+1}),$$
where $c \in \{0,1\}^{n+1}$ is as defined above. It follows from \eqref{eq:delta-maj} that $Q_{\ge w}(x_1, \ldots, x_n)$ $\frac{1}{40n}$-approximates $\Delta_w(x_1, \ldots, x_n)$. Furthermore, $$deg(Q_{\ge w}) = o\left( \sqrt{\frac{ n}{\log n}}\right).$$ Similarly, we can obtain a symmetric torus polynomial $Q_{\ge w+1}(x_1, \ldots, x_n)$ that $\frac{1}{40n}$-approximates $\Delta_{\ge w+1}(x_1, \ldots, x_n)$ such that
$$deg(Q_{\ge w+1}) = o\left( \sqrt{\frac{ n}{\log n}}\right).$$
Note that
$$ \frac{\Delta_w(x)}{2} \pmod 1 = \left(\frac{\Delta_{\ge w}(x)}{2} - \frac{\Delta_{\ge w+1}(x)}{2}\right) \pmod 1.$$
Defining $Q_w(x) = Q_{\ge w}(x) - Q_{\ge w+1}(x) \modone$, it follows that
$$\left\| \frac{\Delta_w(x)}{2} - Q_w(x) \modone\right\|_\infty \le \frac{1}{20n}.$$
This completes the proof.
\end{proof}
The main result of this section now follows from \Cref{thm:delta-lb} and \Cref{lem:maj}:
\begin{corollary}
\label{cor:maj}
Any symmetric torus polynomial of degree $d$ that $\frac{1}{20n}$-approximates $\mathrm{Maj}_n(x)$ must satisfy $d = \Omega\left( \sqrt{\frac{ n}{\log n}}\right)$.
\end{corollary}

\section{Upper bound for delta functions}
\label{sec:upper}
In this section, we prove the somewhat surprising result that if the approximation parameter $\eps>0$ is not too small (say, $\eps$ is a small constant), then the delta function $\Delta_w$ can be nontrivially approximated by \textit{symmetric} low-degree torus polynomials.

\begin{lemma}
\label{lemma:delta}
For every $0 \le w \le n$ and $\eps>0$, there is a symmetric torus polynomial of degree $\frac{\polylog(n/\eps)}{\eps}$ that $\eps$-approximates $\Delta_w(x)$, and thus
$$
\adeg_{\eps}(\Delta_w) \le \frac{\polylog(n/\eps)}{\eps}.
$$
\end{lemma}

\begin{proof}
For any prime $p \ge 2$, let $f_p:\{0,1\}^n \to \{0,1\}$ denote the function
$$
f_p(x) =
\begin{cases}
1 & |x| \equiv w \pmod{p}\\
0 & \text{otherwise}
\end{cases}\;.
$$
It is computed by the $\F_p$-polynomial of degree $p-1$
$$
f_p(x) = 1-\left(\sum x_i - w\right)^{p-1} \pmod{p}.
$$

Let $\cP=\{p_1,\ldots,p_t\}$ be the first $t$ primes, for $t$ to be chosen later. Applying \Cref{lemma:poly_Fp_variant} with $\alpha=1/2t$ and error $\eps/2t$,
for each $p \in \cP$ we obtain a torus polynomial $Q_p:\{0,1\} \to \T$ of degree $O(p\log(t/\eps))$ such that
$$
\left\|Q_p - \frac{1}{2t} f_p \modone \right\|_{\infty} \le \frac{\eps}{2t}.
$$
Define
$$
Q(x) = \sum_{p \in \cP} Q_p(x) \modone.
$$
We claim that $Q$ is a symmetric torus polynomial that $\eps$-approximates $\Delta_w$.

Consider first $x \in \{0,1\}^n$ with $|x|=w$.
In this case, for each $p \in \cP$ we have $f_{p}(x)=1$, $|Q_{p}(x) - \frac{1}{2t} \modone| \le \eps/2t$ and hence
$$
\left|Q(x) - \frac{1}{2} \modone\right| \le \eps/2.
$$
Next, assume that $|x| \ne w$. Then $f_p(x)=1$ only if $p$ divides $|x|-w$. As there are at
most $\log n$ such primes, we have that
$$
\left|Q(x) \modone\right| \le \frac{\eps}{2} + \frac{\log n}{t}.
$$
To conclude we choose $t=O(\log(n)/\eps)$. The largest prime in $\cP$ has size $O(t \log t)$ which means that
$$
\adeg_{\eps}(f) \le \deg(Q) = \max\{\deg(Q_p): p \in \cP\} \le O(t \log t \cdot \log(t/\eps)) = \frac{\polylog(n/\eps)}{\eps}.
$$
To see why $Q$ is symmetric, observe that \Cref{lemma:poly_Fp_variant} preserves symmetry.
\end{proof}

\paragraph{Acknowledgements.} We thank Marco Carmosino for useful discussions regarding this work. We would also like to thank Eric Allender and Sivakanth Gopi for helpful feedback on earlier drafts of this work.

\bibliographystyle{alpha}
\bibliography{torus}

\end{document}